\documentclass[review]{elsarticle}

\usepackage[fleqn]{amsmath}
\usepackage{mathtools}

\usepackage{relsize}

\usepackage{amssymb, latexsym, amsmath}

\usepackage{overpic}
\usepackage{color}

\usepackage{url}

\usepackage{graphicx}
\usepackage{tikz}
\usetikzlibrary{arrows,automata}

\usepackage{amsfonts}

\usepackage{epstopdf}

\usepackage[font=small, labelfont=bf]{caption}
\usepackage{caption}
\usepackage{sidecap}
\usepackage{subcaption}
\usepackage{adjustbox}
\usepackage{epsfig}
\graphicspath{{../Plots/}}

\usepackage{pdflscape}

\newtheorem{thm}{Theorem}
\newtheorem{defn}{Definition}
\newtheorem{definition}{Definition}[section]

\newtheorem{lem}{Lemma}
\newtheorem{rem}{Remark}

\newenvironment{proof}{\noindent{\bf Proof:}}{$\hfill \Box$ \vspace{10pt}}

\begin{document}

\begin{frontmatter}

\title{A new fractional model in Caputo sense for studying the dynamics of COVID-19 spread in France}
\author[1]{Mahmoud H. A. Saleh\corref{cor1}}
\ead{mahmoud.h.a.saleh@gmail.com}
\author[2]{Tarek M. Abed-Elhameed}
\ead{tarekmalsbagh@aun.edu.eg}
\cortext[cor1]{corresponding author}
\address[1]{Department of Physics, Faculty of Science, Assiut University, Assiut 71516, Egypt.}
\address[2]{Department of Mathematics, Faculty of Science, Assiut University, Assiut 71516, Egypt.}
\date{Received: date / Accepted: date}

\begin{abstract}
The COVID-19 pandemic has rapidly spread around the world and burdened public health in almost all countries involving France. After the spread of SARS-CoV-2, France harvested many deaths in total. In this paper, we develop models with integer and fractional orders to investigate the dynamics of COVID-19 transmission in French hospitals and intensive care units (ICUs). Moreover, this paper aims to explore the impact of precautionary measures on the total infected cases in hospitals and ICUs of COVID-19 for the entire France by using available actual data.

\end{abstract}
\end{frontmatter}

	
\section{Introduction}
{The novel coronavirus has been commonly known as a severe acute respiratory syndrome novel coronavirus 2 (SARS-CoV-2) or coronavirus disease 2019 (COVID-19). The pandemic onset began on 31 December 2019 when the first case was detected in Wuhan city, Hubei Province of China~\cite{world2020situation}. From January 2020, the number of COVID-19 cases has increased rapidly. The SARS-CoV-2  has affected 213 countries and territories around the world besides two international conveyances. By 27 July 2020, the total cumulative COVID-19 cases reached 16,638,667, and the total number of deaths became 656,924~\cite{worldometer2020covid}.
	
On 24 January 2020, the virus was confirmed to have reached France where the first three cases of COVID-19 were identified. Two cases of them were diagnosed in Paris and one in Bordeaux~\cite{stoecklin2020first, lefigaro2020france}. But it appeared to us, subsequently, that there were already some cases exhibiting in December 2019~\cite{kamdar2020return}. Also, the basic reproduction number value in France was estimated to $R_0 \approx 2.96$ as in~\cite{salje2020estimating} and $R_0 \approx 3.2$ in~\cite{chen2020mathematical} before lockdown on 17 March 2020. After lockdown in France, we found that the effective in reducing the spread of COVID-19 appeared and the effective reproduction number, called $R_t$, divided by a factor 5–7 at the country scale by 11 May~\cite{roques2020impact, salje2020estimating}.
	
In general, intensive-care units (ICUs) treated highly miscellaneous patients at high peril of death-rate~\cite{nielsen2019survival}. Previous studies~\cite{pandharipande2013long, tansey2007one, bienvenu2018psychiatric, fan2014physical} of many critically ill patients who survived from intensive care suffered from subtending impaired outcomes in long-lasting physical, cognitive and/or mental health. In addition, the study~\cite{korupolu2020rehabilitation} recommended the rehabilitation of critically ill COVID-19 survivors. The study~\cite{kamdar2020return} showed that the survivors from ICU hospitalization showed a delay in returning to work, resulting in substantial economic sequels. Another study~\cite{armstrong2020outcomes} showed that the mortality rate of COVID-19 patients in ICUs, across the international studies, was found to be $41.6 (34.0 - 49.7)\%$ \, ($95\%$CI).
	
Mathematical modeling has a significant effect on providing simulating for the epidemic diseases and other several complex phenomena that vividly appeared in the literature such as~\cite{naik2020modeling, owolabi2020mathematical}. Many papers, such as~\cite{kucharski2020early, chen2020mathematical, ivorra2020mathematical} \cite{ndairou2020mathematical, ngonghala2020mathematical}, were published to analyze the spread of SARS-CoV-2 using mathematical models. Also, some papers dealt with mathematical modeling of COVID-19 involving hospitalization class as in~\cite{nabi2020forecasting, benlloch2020effect} and including critically-infected cases such as~\cite{ullah2020modeling}.
	
In recent years, with the continued development of fractional calculus theory and fractional-order system modeling approaches, one can see how important it is in various engineering and other fields \cite{cao2018fractional, west2016fractional}. Besides that, many scholars have researched COVID-19 in the context of fractional mathematical modeling~\cite{xu2020forecast, nisar2020mathematical, khan2020dynamics}. Fractional derivatives have more accuracy in describing the behavior of biological systems. And that returns to it has the information history of the whole time interval or long memory than the short one of integer derivatives.

This manuscript is organized in the following manner:
\begin{itemize}
	\item Sec.~\eqref{sec100}, we present the equations of the integer-order SIHUR model and its initial conditions.
	\item  In Sec.~\eqref{sec101}, we study the essential analytical procedures such as positivity and boundedness, in subsection~\eqref{sec102}, in addition to the invariant region set in~\eqref{sec103}.
	\item Sec.~\eqref{sec104}, we propose the fractional-order SIHUR model in Caputo sense together with the initial conditions~\eqref{sec105}.
	\item With Sec.~\eqref{sec106}, we determine the disease-free equilibrium point (DFEP) to get the basic reproduction number $R_0$ and study its stability as in subsections~\eqref{sec107} and~\eqref{sec108}.
	\item In Sec.~\eqref{sec109}, sensitivity analysis and estimation of the SIHUR parameters are determined.
	\item The last two remaining sections~\eqref{sec110} and~\eqref{sec111}, and we discuss the numerical simulations and conclude the core of this study, respectively.
\end{itemize}
\section{Model formulation}\label{sec100}
We designed our SIHUR model of SARS-CoV-2 spread, which corresponds to the data collection, as shown in the following equations.
\ref{f100}.
\begin{figure}[h]
\begin{center}
\centering
\includegraphics[scale=0.5]{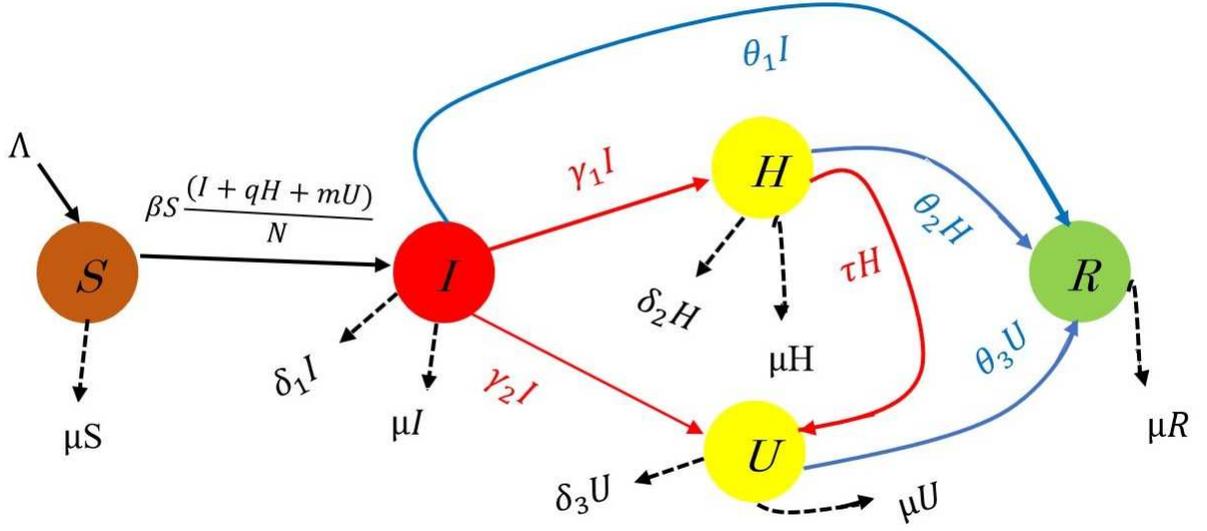}\\
\caption{SIHUR model.}
\label{f100}
\end{center}
\end{figure}
\begin{center}
\centering
\large
\begin{align}
\frac{dS}{dt} &=\Lambda -\varphi S-\mu S \nonumber,\\
\nonumber \\
\frac{dI}{dt} &=\varphi S - k_1 I  \nonumber,\\
\nonumber \\
\frac{dH}{dt} &=\gamma_1 I- k_2 H \label{E099},\\
\nonumber \\
\frac{dU}{dt} &=\gamma_2 I+\tau H- k_3 U \nonumber,\\
\nonumber \\
\frac{dR}{dt} &=\theta_1 I+\theta_2 H+\theta_3 U- \mu R. \nonumber
\end{align}
\end{center}
Subject to the initial conditions to the model~\eqref{E099}:\\
$S(0)\geq0,~I(0)\geq0,~H(0)\geq0,~U(0)\geq0,~R(0)\geq0$.\\
Where:\\
 $\varphi = \frac{\beta (I+ q H + m U)}{N}$, $k_1=(\gamma_1 +\gamma_2+\theta_1+\delta_1+\mu)$, $k_2=(\tau +\theta_2+\delta_2+\mu)$ and $k_3=(\theta_3 +\delta_3+ \mu)$.\\
 
 \begin{enumerate}
	\item The description of the divided classes of the whole human population$N(t)$, where\\
	$N(t)=S(t)+I(t)+U(t)+H(t)+R(t)$, is as below:
	\begin{itemize}
		\item The susceptible individuals density class is S(t).
		\item The infected class with the COVID-19 symptoms is I(t).
		\item The hospitalized patients class is H(t).
		\item The patients in intensive care units (ICUs) or (critically infected cases) class is U(t).
		\item The recovered$/$removed individuals class is R(t).
	\end{itemize}
	
	\item The description of the parameters:
	\begin{itemize}
		\item $\Lambda$ is the birth rate.
		\item $\mu$ is the natural death rate in all classes.
		\item $\theta_1$, $\theta_2$ and $\theta_3$ represent the recovery rates of all infected cases in the class $I$, hospitalized (in-hospital) patients in the class $H$ and in-ICUs patients in the class $U$ respectively.
		\item $\delta_1$, $\delta_2$ and $\delta_3$ are the mortality rates due to the COVID-19 in the $I$, $H$ and $U$, respectively.
		\item $\beta$ shows the transmission rate of the infectious disease, COVID-19.
		\item $\gamma_1$ denotes the transmissibility rate of the infected people with COVID-19 who need hospitalization.
		\item $\gamma_2$ denotes the transmissibility rate of the infected people with COVID-19 who need the ICUs.
		\item $\tau$ denotes the transmissibility rate of the hospitalized people with COVID-19 to the ICUs.
		\item $q$ is the infection rate due to health care workers (HCWs) interacting with hospitalized COVID-19 patients in hospitals.
		\item $m$ hints at the infection rate because of contact between the HCWs and the critical care patients in ICUs.
	\end{itemize}
\end{enumerate}

\begin{rem}
	We assume $q$ and $m$ represent the infection rate parameters between HCWs and infected patients in hospitals and ICUs, respectively. That interaction between them generates new COVID-19 ones among the HCWs. Then, this possibly generates further new cases out of hospitals and ICUs.
\end{rem}

\section{Some basic analytical results}\label{sec101}
Herein, we present some necessary analytical results of the SIHUR model~\eqref{E099} of COVID-19, such as positivity and boundedness for solutions of our model, and stability of the disease-free equilibrium point (DFEP) after presenting the fractional model. Moreover, we complete the basics by establishing the theoretical formula of the crucial biological parameter called the basic reproduction number, $R_0$, of the fractional SIHUR model.

\subsection{Positivity and boundedness of the solution}\label{sec102}
For the first part, we examine the solution positivity of our model.
\begin{lem}\label{lem100}
	Let the initial data be $P(0) \geq 0$ and $P(t) = (S, I, H, U, R)$ are the variables of the model. Then, all the solutions of the model~\eqref{E099} will be non-negative for all $t > 0$. Further,\\
	\begin{equation*}
	\lim_{t \to \infty} \sup_{}N(t) \leq \frac{\Lambda}{\mu}.
	\end{equation*}
\end{lem}

\begin{proof}{
		Let $t_{*} = \sup_{} \{t > 0 : P(t) > 0 \in [0, t] \}$. Then, it follows the first Eq of our model written as shown:\\
		\begin{equation}\label{Epostive101}
		\frac{dS}{dt} =\Lambda-(\varphi + \mu) S.
		\end{equation}\\
		Multiply both sides of the above Eq.~\eqref{Epostive101} by the integrating factor, $\mathrm{exp}\left(\mu t + \int_{0}^{t} \varphi(\zeta) d\zeta \right)$. Then, it can be further written as below,\\
		\begin{equation*}
		\frac{d}{dt} \left\{ S(t)\times \mathrm{exp}\left(\mu t + \int_{0}^{t} \varphi(\zeta) d\zeta \right) \right\} = \Lambda\times \mathrm{exp}\left( \mu t + \int_{0}^{t} \varphi(\zeta) d\zeta \right).
		\end{equation*}\\
		In consequence,\\
		\begin{equation}\label{E107}
		S(t_{*})\times \mathrm{exp}\left(\mu t_{*} + \int_{0}^{t_{*}} \varphi(\zeta) d\zeta \right) - S(0) = \Lambda \int_{0}^{t_{*}}  \mathrm{exp}\left(\mu t + \int_{0}^{t} \varphi(\nu) d\nu \right) dt.
		\end{equation}\\
		Then, the solution of Eq. (\ref{E107}) is\\
		\begin{align}
		S(t_{*}) = S(0)&\times \mathrm{exp}\left\{ - \left(\mu t_{*} + \int_{0}^{t_{*}} \varphi(\zeta) d\zeta \right) \right\} + \mathrm{exp}\left\{ - \left(\mu t_{*} + \int_{0}^{t_{*}} \varphi(\zeta) d\zeta \right) \right\}\\
		&\times \Lambda \int_{0}^{t_{*}}  \mathrm{exp}\left(\mu t + \int_{0}^{t} \varphi(\nu) d\nu \right) dt > 0. \nonumber
		\end{align}\\
		Doing similar procedure for the rest of the Eqs. of the model~\eqref{E099}, it can be proved that $P(t)  > 0$, where $\forall t > 0$.\\
		
		For examining the second part, the boundedness property, we have $0 < P(0) \leq N(t)$. Adding all the equations~\eqref{E099} of the system , we get\\
		\begin{align}
		\frac{dN(t)}{dt} &= \Lambda - \mu N(t) - \delta_1 I - \delta_2 H - \delta_3 U \nonumber \\ 
		&\leq \Lambda - \mu N(t). \nonumber
		\end{align}
		Hence, \\
		\begin{equation*}
		\lim_{t \to \infty} \sup_{}{N(t)} \leq \frac{\Lambda}{\mu}.
		\end{equation*}
	}
\end{proof}

\subsubsection{Invariant region}\label{sec103}
Over here, we will study the dynamics of the COVID-19 mathematical model in the following closed biologically feasible region.
\begin{center}
	\begin{equation*}
	\Upsilon \subset \mathbb{R}_{+}^5,
	\end{equation*}
\end{center}
where,\\
$\Upsilon = \left\{(S, I, H, U, R)\in\mathbb{R}_{+}^5: S+ I+ H+ U+R \leq \frac{\Lambda}{\mu} \right\}$.\\
\begin{lem}
	The region given by $\Upsilon$, where $\Upsilon \subset \mathbb{R}_{+}^5$, is positively invariant set for the model~\eqref{E099} with non-negative initial conditions in $\mathbb{R}_{+}^5$
\end{lem}

\begin{proof}{
		As in Lemma (\ref{lem100}), it follows from the adding of all equation of the human population of the model~\eqref{E099} \\
		\begin{equation}\label{E106}
		\frac{dN(t)}{dt} \leq \Lambda - \mu N(t).
		\end{equation}\\
		So, it is vividly that\\
		\begin{equation*}
		\frac{dN(t)}{dt} \leq 0,~ if~ N(0) \geq \frac{\Lambda}{\mu}.
		\end{equation*}\\
		And the solution of (\ref{E106}) is given by the following inequality,\\
		\begin{equation*}
		N(t) \leq N(0)e^{- \mu t} + \frac{\Lambda}{\mu} (1- e^{- \mu t}).
		\end{equation*}\\
		Then, $N(t) \leq \frac{\Lambda}{\mu}$, if $N(0) \leq \frac{\Lambda}{\mu}$. Thus, the region $\Upsilon$ is positively invariant set and also attract all the possible solution trajectories in $\mathbb{R}_{+}^5$.
	}
\end{proof}
\section{Mathematical formulation of the fractional SIHUR model of COVID-19}\label{sec104}
In this section, we discuss the basic literature related to the fractional operator and its applications to the SIHUR model of COVID-19.
\subsection{Preliminaries of fractional calculus}\label{sec105}
Some useful definitions and lemmas are shown as follows
\begin{defn}\cite{podlubny1998fractional}
	The definition of Caputo fractional-order derivative is given by
	\begin{equation}\label{E109}
	\mathrm{}_{0}^{C}\mathrm{D}_{t}^{\alpha} f(t) = 
	\begin{cases}
	\text{$\frac{1}{\Gamma (\nu-\alpha)} \int_{0}^{t}\frac{f^{(\nu)}(\xi)}{(t-\xi)^{1+\alpha -\nu}}d\xi$,} &\quad\text{$\nu-1 < \alpha < \nu$,}\\
	\text{$\frac{d^{\nu}}{dt^{\nu}} f(t)$,} &\quad\text{$\alpha =\nu$}, \\
	\end{cases}
	\end{equation}
	where $\nu$ presents the smallest positive integer not less than $\alpha$.
\end{defn}
\begin{defn}\cite{podlubny1998fractional, li2009mittag}
	The Mittag–Leffler function $E_{\alpha}(Z)$ is given by
	\begin{equation*}
	E_{\alpha}(Z) = \sum_{\varkappa =0}^{\infty} \frac{Z^{\varkappa}}{\Gamma (\varkappa \alpha +1)},
	\end{equation*}
	where $\nu-1 \leq \alpha \leq \nu$
\end{defn}
\subsection{SIHUR model}\label{sec106}
We present the dynamics of the COVID-19 model~\eqref{E099} in Caputo sense~\eqref{E109} and we have the following model:
\begin{center}
	\centering
	\large
	\begin{align}
	\mathrm{}_{0}^{C}\mathrm{D}_{t}^{\alpha} S(t) &=\Lambda -\varphi S-\mu S \label{E100},\\
	\nonumber \\
	\mathrm{}_{0}^{C}\mathrm{D}_{t}^{\alpha} I(t) &=\varphi S - k_1 I \label{E101},\\
	\nonumber \\
	\mathrm{}_{0}^{C}\mathrm{D}_{t}^{\alpha} H(t) &=\gamma_1 I- k_2 H \label{E102},\\
	\nonumber \\
	\mathrm{}_{0}^{C}\mathrm{D}_{t}^{\alpha} U(t) &=\gamma_2 I+\tau H- k_3 U \label{E103},\\
	\nonumber \\
	\mathrm{}_{0}^{C}\mathrm{D}_{t}^{\alpha} R(t) &=\theta_1 I+\theta_2 H+\theta_3 U- \mu R. \label{E104}
	\end{align}
\end{center}

The initial conditions for the model of Eqs.~(\ref{E100}-\ref{E104}) are:\\
$S(0)\geq0,~I(0)\geq0,~H(0)\geq0,~U(0)\geq0,~R(0)\geq0$,\\where:\\
$\alpha$ represents the fractional-order,
$\varphi = \frac{\beta (I+ q H + m U)}{N}$, $k_1=(\gamma_1 +\gamma_2+\theta_1+\delta_1+\mu)$, $k_2=(\tau +\theta_2+\delta_2+\mu)$ and $k_3=(\theta_3 +\delta_3+ \mu)$.\\
\section{The basic reproduction number}\label{sec107}
The basic reproduction number, which is known as $R_0$, is defined as the predictable number of secondary cases generated by an ideal infected individual in a totally susceptible population~\cite{heesterbeek2002brief}. We calculate $R_0$ via next-generation matrix approach, see~\cite{diekmann1990definition,van2002reproduction}. In order to calculate $R_0$, we first need to determine the disease-free equilibrium point (DFEP) point by setting the right-hand side of the Eqs. (\ref{E100}-\ref{E104}) equal to zero and obtain
\begin{center}
	\centering
	\large
	\begin{align}
	\mathrm{}_{0}^{C}\mathrm{D}_{t}^{\alpha} S(t) &=\Lambda -\varphi S-\mu S=0 \nonumber,\\
	\mathrm{}_{0}^{C}\mathrm{D}_{t}^{\alpha} I(t) &=\varphi S - k_1 I=0 \nonumber,\\
	\mathrm{}_{0}^{C}\mathrm{D}_{t}^{\alpha} H(t) &=\gamma_1 I- k_2 H=0 \nonumber,\\
	\mathrm{}_{0}^{C}\mathrm{D}_{t}^{\alpha} U(t) &=\gamma_2 I+\tau H- k_3 U=0 \nonumber,\\
	\mathrm{}_{0}^{C}\mathrm{D}_{t}^{\alpha} R(t) &=\theta_1 I+\theta_2 H+\theta_3 U- \mu R=0. \nonumber
	\end{align}
\end{center}
Then, the DFEP is the point in which there does not exist disease among all individuals in the studied population, so the disease-free equilibrium point is given by\\ $D_0 = (S, I, H, U, R) = (\frac{\Lambda}{\mu}, 0, 0, 0, 0)$.\\
The basic reproduction number is very crucial for the qualitative analysis of our model. We compute $R_0$ by the following steps: \\

First, we divide our system into infected class Eqs. ($\ref{E101}-\ref{E103}$) and non-infected class Eqs. \eqref{E100} and \eqref{E104}. Next,
\begin{center}
	\begin{equation}
	\frac{d}{dt}  \begin{bmatrix}
	I \\ H \\ U
	\end{bmatrix} = \psi - \zeta,
	\end{equation}
\end{center}
where 

$\psi = \begin{bmatrix}
\frac{\beta S(I+ q H + m U)}{N}  \\ 0 \\ 0
\end{bmatrix}$ and
$\zeta = \begin{bmatrix}
k_1 I \\ -\gamma_1 I + k_2 U \\ -\gamma_2 I - \tau H + k_3 U
\end{bmatrix}$.\\
Then,\\
$F=J_{D_0}(\psi) = \begin{bmatrix}
\beta & \beta q & \beta m \\ 0 & 0 & 0 \\ 0 & 0 & 0
\end{bmatrix}, $
$V=J_{e_0}(\zeta) = \begin{bmatrix}
k_2 k_3& 0 &0 \\ \gamma_1 k_3 & k_1 k_3 &0 \\ (\gamma_1 \tau +\gamma_2 k_2) &k_1 \tau & k_1 k_2
\end{bmatrix},$\\
and
$ V^{-1}=(J_{D_0}(\zeta))^{-1} = \begin{bmatrix}
\frac{1}{k_1} & 0&0\\ \frac{\gamma_1}{k_1 k_2} & \frac{1}{k_2}&0 \\ \frac{(\gamma_1 \tau + \gamma_2 k_2)}{k_1 k_2 k_3} &\frac{1}{k_2 k_3}& \frac{1}{k_3}
\end{bmatrix}.$\\

Thereafter, the basic reproduction number is determined from the formula $R_0 \coloneqq \rho[FV^{-1}]$. Where $\rho(.)$  is the spectral of the matrix $F=J_{D_0}(\psi)$ and $V^{-1}=(J_{D_0}(\zeta))^{-1}$.  $J_{D_0}$ represents the Jacobian matrix of the matrices $\psi$ and $\zeta$ at DFEP, $D_0$. $F V^{-1}$ is the next generation matrix. Accordingly, $R_0$ can be written afterwards:\\
\begin{equation}
R_0=\frac{\beta}{k_1} \bigg(1+q\frac{\gamma_1}{k_2}+m\frac{(\gamma_1 \tau + \gamma_2 k_2)}{k_2 k_3} \bigg).
\label{E108}
\end{equation}
Further, we can write $R_0$ in the following form:\\
\begin{equation*}
R_0 = R_1 + R_2 + R_3 +R_4 ,
\end{equation*}
where,
\begin{equation*}
R_1 = \frac{\beta}{k_1},~~~~~ R_2 = \beta q\frac{\gamma_1}{k_1 k_2}, ~~~~~ R_3=\beta m\frac{\gamma_1 \tau}{k_1 k_2 k_3},~~~~~ R_4= \beta m \frac{ \gamma_2}{k_1 k_3}.
\end{equation*}

\subsection{Stability of disease-free equilibrium point}\label{sec108}
In this section, we calculate both local and global stability of our system (\ref{E100}-\ref{E104}), see section~\ref{sec100}, around the disease-free equilibrium point (DFEP). We then prove the local and global stability of our proposed model around the DFEP. The epidemiological effect of the stability of the output of the DFEP state is when a small influx of COVID-19 infections cases does not cause a COVID-19 outbreak in the case $R_0<1$.

\subsubsection{Local stability of disease-free equilibrium, $D_0$}\label{sec109}
\begin{thm}
	If $R_0<1$, then the DFEP, $D_0$ of the model (\ref{E100}-\ref{E104}) is locally asymptotically stable if the condition $\mid arg(\lambda) \mid \textgreater \alpha \frac{\pi}{2}$ is satisfied and unstable otherwise.
\end{thm}
\begin{proof}
	The Jacobian matrix $J_{D_0}$ obtained at the DFE $D_0$ for our system (\ref{E100}-\ref{E104}) as: 
	\begin{equation}\label{mat1}
	J_{D_0} = \begin{bmatrix}
	-\mu & -\beta & -\beta q & -\beta m & 0 \\
	0 & \beta-k_1 & \beta q & \beta m & 0 \\
	0 & \gamma_1 & -k_2 & 0 & 0 \\
	0 & \gamma_2 & \tau & -k_3 & 0 \\
	0 & \theta_1 & \theta_2 & \theta_3 & -\mu \\
	\end{bmatrix},
	\end{equation}
	Vividly, from the matrix \eqref{mat1} the eigenvalues, $-\mu$ and$-\mu$ are negative values. The remaining eigenvalues obtained from the below equation:\\
	\begin{equation}\label{mat2}
	\lambda^3+C_1 \lambda^2+C_2 \lambda+C_3 =0.
	\end{equation}
	The coefficients involved in Eq. \eqref{mat1} are as follow:\\
	\begin{align*}
	C_1&=k_1 + k_2 + k_3 - k_1 R_1,\\
	C_2&= k_2 k_3 +k_1 k_2 (1-R_2)+ k_1 k_3 (1-R_4) - k_1 (k_2 + k_3) R_1 ,\\
	C_3&=k_1 k_2 k_3 (1-R_0).
	\end{align*}
	
	Vividly,  $\forall~ C_i$, where $i=1, 2, 3$ are positive if $R_0 < 1$. Then, the eigenvalues from Eq. \eqref{mat2} are negative or have a negative real part. Since all the eigenvalues $\lambda_i$, ($i = 1, 2, ..., 6$), of \ref{mat1} satisfy the condition $\mid arg(\lambda) \mid \textgreater \alpha \frac{\pi}{2}$, the DFEP ($D_0$) is locally asymptotically stable.
\end{proof}

\subsubsection{Global stability of disease-free equilibrium, $D_0$}\label{sec110}
The global stability of the system at the DFEP, $D_0$, of the COVID-19 transmission model is investigated in the following consequence:
\begin{thm}
	The non-linear dynamical system (\ref{E100}-\ref{E104}) at the DFEP ($D_0$) is globally asymptotically stable in $\Upsilon$ when $R_0<1$, and unstable in the case $R_0>1$.
\end{thm}
\begin{proof}{
		Let we suppose the following Lyapunov function to prove the desired consequence:\\
		\begin{equation*}
		L(t) = \Psi_1 I + \Psi_2 H + \Psi_3 U,
		\end{equation*}
		where $\Psi_j$, for $j=1, 2, 3$, are used for unknown positive constants. By differentiating the function $L(t)$ with respect to time $t$ and using the Eqs. (\ref{E101}-\ref{E102}), we get:
		\begin{align*}
		\frac{dL(t)}{dt}&= \Psi_1 \bigg\{ \frac{\beta S(I+ q H + m U)}{N}-(\gamma_1 +\gamma_2+\theta_1+\delta_1+\mu) I \bigg\}+\Psi_2 \bigg\{\gamma_1 I-(\tau +\theta_2+\delta_2+\mu) H \bigg\}\\ &+\Psi_3 \bigg\{\gamma_2 I+\tau H-(\theta_3 +\delta_3+ \mu) U \bigg\}.\\
		\end{align*}
		As $\frac{S}{N}<1$. Then, \\
		\begin{align*}
		\frac{dL(t)}{dt}&\leq \Psi_1 \bigg\{ \beta (I+ q H + m U)-(\gamma_1 +\gamma_2+\theta_1+\delta_1+\mu) I \bigg\}+\Psi_2 \bigg\{\gamma_1 I-(\tau +\theta_2+\delta_2+\mu) H \bigg\}\\ &+\Psi_3 \bigg\{\gamma_2 I+\tau H-(\theta_3 +\delta_3+ \mu) U \bigg\} = (\Psi_1 \beta-\Psi_1 \gamma_1-\Psi_1 \gamma_2 -\Psi_1 \theta_1-\Psi_1 \delta_1-\Psi_1 \mu+\Psi_2 \gamma_1+\Psi_3 \gamma_2) I\\&+(\Psi_1 \beta q-\Psi_2 \tau-\Psi_2 \theta_2-\Psi_2 \delta_2-\Psi_2 \mu+\Psi_3 \tau) H+(\Psi_1 \beta m-\Psi_3 \theta_3-\Psi_3 \delta_3-\Psi_3 \mu) U.
		\end{align*}
		Now, let $\Psi_1$, $\Psi_2$ and $\Psi_3$ values as below, \\
		$\Psi_1= \frac{1}{k_1}$, $\Psi_2=\beta m\frac{\tau}{k_1 k_2 k_3}$ and $\Psi_3= \beta m\frac{1}{k_1 k_3}$. \\
		Then, we get,
		\begin{equation*}
		\frac{dL(t)}{dt}\leq  (R_1 + R_3 + R_4 -1) I+ q (R_2) H.
		\end{equation*}
		Obviously, $\frac{dL(t)}{dt} < 0$ when $R_0 < 1$. Thus, the the largest compact invariant set in $\Upsilon$ is the singleton set $D_0$ according to the LaSalle’s invariant principle \cite{la1976stability}. Thus, $D_0$ is entirely globally asymptotically stable in $\Upsilon$.}
\end{proof}

\section{Sensitivity analysis and estimation of SIHUR parameters}\label{sec111}
In the previous Section~\ref{sec101}, we derived the basic reproduction number for the system (\ref{E100}-\ref{E104}) proposed in Section~\ref{sec100}. The sensitivity analysis for the threshold of pandemic ($R_0$) is significant. It is usually used to determine the robustness of the model predictions to parameter values, since there are common errors or uncertainty in data collection and assumed parameter values~\cite{saltelli2008global, chitnis2008determining, mcleod2006sensitivity}. The importance of sensitivity indices arises in measuring the relative change in a variable when a parameter value changes. Therefore, here we use the normalized forward sensitivity index of a variable with respect to a given parameter. This is defined as the ratio of the relative change in some variable to the relative change in the given parameter. When the variable is a differentiable function of the parameter, the sensitivity index maybe, on the other hand, be defined as the partial derivative of the variable with respect to the given parameter. We can write the definition as follows.
\begin{definition}\cite{chitnis2008determining}:
	The normalized forward sensitivity index of a variable, $R_0$, that depends differentiably on a given parameter, $\chi$, which is defined as:
	\begin{equation}
	\Omega_{\chi}^{R_0} = \frac{\partial R_0}{\partial \chi} \times \frac{\chi}{R_0}
	\label{e112}
	\end{equation}
\end{definition}
 The sensitivity indices, sources, parameters values and initial values are arranged in tables \ref{t100}-\ref{t102}.
\begin{table}[h]
	\centering
	\begin{tabular}{|c|c|c|}
		\hline
		\multicolumn{3}{|c|}{France} \\ \hline
		Symbol & Initial value & Method/ Source  \\ \hline
		$S(0)$  & 19948051 & $S(0) =N(0) - (I(0) +H(0) +U(0) +R(0))$ \\ \hline
		$I(0)$   & 29440000   & ~\cite{owidcoronavirus} \\ \hline
		$H(0)$  & 16785   & ~\cite{owidcoronavirus} \\ \hline
		$U(0)$  & 1074  & ~\cite{owidcoronavirus}\\ \hline
		$R(0)$  & 15867601  & ~\cite{numberscovid19france} \\ \hline
	\end{tabular}
	\caption{Initial values of the SIHUR model for France.}
	\label{t100}
\end{table}

\begin{table}[h]
	\centering
	\begin{tabular}{|c|c|c|c|}
		\hline
	 \multicolumn{4}{|c|}{France}                                            \\ \hline
		Parameter &Value &Sensitivity index&Method/ Source  \\ \hline
		$\mu$  & $\frac{1}{83.13 \times 365}$ & $-6.5048\times 10^{-4}$ & ~\cite{france2020population} \\ \hline
		$\Lambda$  &  $0.011$ & $-$ & ~\cite{ined2021population}  \\ \hline
		$\beta$ &$0.204$ & 1 &~\cite{sharma2022parameter} \\ \hline
		$\delta_1$ & $5.027\times 10^{-3}$   & $   -0.0992
		$ & \cite{owidcoronavirus}  \\ \hline
		$\delta_2$ & $0.2$ &$-2.33\times 10^{-5}$& ~\cite{boelle2020trajectories} \\ \hline
		$\delta_3$ & $0.416$  &$-1.4257\times 10^{-5}$& ~\cite{armstrong2020outcomes} \\ \hline
		$\theta_1$ & $0.045$  &$-0.8882$ & Assumption \\ \hline
		$\theta_2$ & $0.8$ & $-9.3242\times 10^{-5}$ & ~\cite{prescott2020recovery}   \\ \hline
		$\theta_3$ & $ 0.6 $ & $ -2.0564\times 10^{-5} $ & ~\cite{prescott2020recovery} \\ \hline
		$\gamma_1$ & $5.701 \times 10^{-4}$ &$-0.0111$ &~\cite{owidcoronavirus} \\ \hline
		$\gamma_2$ & $ -7.3793\times 10^{-4}$  & $-7.0205\times 10^{-4}$ &~\cite{owidcoronavirus}  \\ \hline
		$\tau$ & $0.064$ & $9.4130\times 10^{-6}$ &~\cite{owidcoronavirus} \\ \hline
		$q$  & $0.2$ &$1.0714\times 10^{-4}$ & Assumption \\ \hline
		$m$  & $0.5$ &$3.4822\times 10^{-5}$ & Assumption  \\ \hline
	\end{tabular}
	\caption{Sensitivity indices, data and values of parameters in the SIHUR model for France.}
	\label{t102}
\end{table}
\begin{rem}
	The natural mortality rate is defined as $\mu = \frac{1}{\emph{Life}~\emph{expectancy} \times 365}$.
\end{rem}

\section{Discussion and numerical results}\label{sec112}
This section was devoted to discussing the numerical simulation results of our integer order system of Eqs.~ (\ref{E100}-\ref{E104}) and fractional-order one~\ref{E109}. We performed the simulation results by using ode45 and fde12 packages in MATLAB, subjected to the parameters and initial values listed in Tables~\ref{t102} and~\ref{t100}, respectively. The number of population was from~\cite{france2020population} for France ($N(0)=65,273,511$). We studied the corresponding dynamics of the coronavirus spread in France for total COVID-19 hospitalized cases and critical care patients in ICUs for the period from $24^{th}$ May 2022 to $13^{th}$ July 2022 (50 days). And we also investigated the impact of precautionary measures, e. g., social distancing and wearing masks, on hospitals and ICUs in France for the same period.

We predicted the COVID-19 development trends as shown in Figs.~\ref{global fig101} with integer and non-integer order. In Fig.~\ref{f111}, we see that the total number of hospitalized cases decreases from $16785$ to nearly $16000$ in the first $4$ days, and then the number of hospitalizations is approximately constant until the $7$ day. After that, we observe a notable exponential decay until the day $50$. We also notice an acceptable match between the two carves for integer and non-integer one. We show the two curves are identical to nearly the day $10$. Next, from day $10$ to day $50$, a suitable match is indeed observed. For Fig.~\ref{f112}, a vast increase in the total number of critical care cases in ICUs is revealed, exceeding $2000$ cases, until the day $5$. Following this is a major decrease in cases of ICUs until the last day of our study. And, again, a satisfactory match between the two carves is noticed.

Precautionary measures are essential measures that several governments imposed. Naturally, these should be directly taken when an epidemic is at the onset. We, here, only study the impact of precautionary measures on the cumulative number of hospitalized and critical care COVID-19 cases, see Figs.~\ref{global fig103} and~\ref{global fig104}, by using our models proposed in Sec.~\ref{sec100} and~\ref{sec104}. We assume that the mitigation in precautionary measures, the transmission rate ($\beta$), increases by $0.05$. So, we have analyzed the impact of basic ($\beta = 0.204$), below moderate ($\beta = 0.254$), moderate ($\beta = 0.304$), upper moderate ($\beta =0.354$) and high ($\beta = 0.404$) precautionary measures.

In Figs.~\ref{f116} and~\ref{f118}, it is observed that mitigation in precautionary has a major effect on increasing the total number of COVID-19 cases in the hospitals for all values of $\beta$. This major effect is for the first $10$ days and the first $15$ days by the integer and fractional order models, respectively. After these periods, a minor effect of the mitigation is shown. We see a rapid decrease proceeds from days $10$ to $50$ in Fig.~\ref{f116} and from days $15$ to $50$ in Fig.~\ref{f118}. In Fig.~\ref{f116}, we find that for different values of $\beta$, all curves intercept in the period from days $24$ to $28$, where the number of cases is nearly $0.8\times 10^{4}$. The same intersection is shown in Fig.~\ref{f118}, but for a period from days $28$ to $32$, where cases are also ranging in $0.8\times 10^{4}$. In addition, for both Figs.~\ref{f116} and~\ref{f118}, after the intersection region, we see all curves very slightly separate from each other.


In Figs.~\ref{f117} and~\ref{f119}, we vividly see that the cumulative number of COVID-19 cases in French ICUs accelerates significantly on account of the decrease in precautionary measures, i. e., increment in the value of $\beta$. That is, where the number of cases surpasses $2000$. This increment is still until the first $5$ and $6$ days for Fig.~\ref{f117} and~\ref{f119}, respectively. And this is followed by a major decrease, continuing to the end of the period in both graphs,~\ref{f117} and~\ref{f119}. The curves of both mentioned two graphs, for different values of $\beta$, are nearly identical for period ($24$-$28$) days in~\ref{f117} and ($28$-$33$) days in~\ref{f119}. From days $28$ (Fig.~\ref{f117}) and $33$ (Fig.~\ref{f119}), for different values of $\beta$, curves are very close together from one another.


\begin{figure}[h]
	\centering
	\begin{subfigure}[ht]{0.7\linewidth}
		\includegraphics[width=\linewidth]{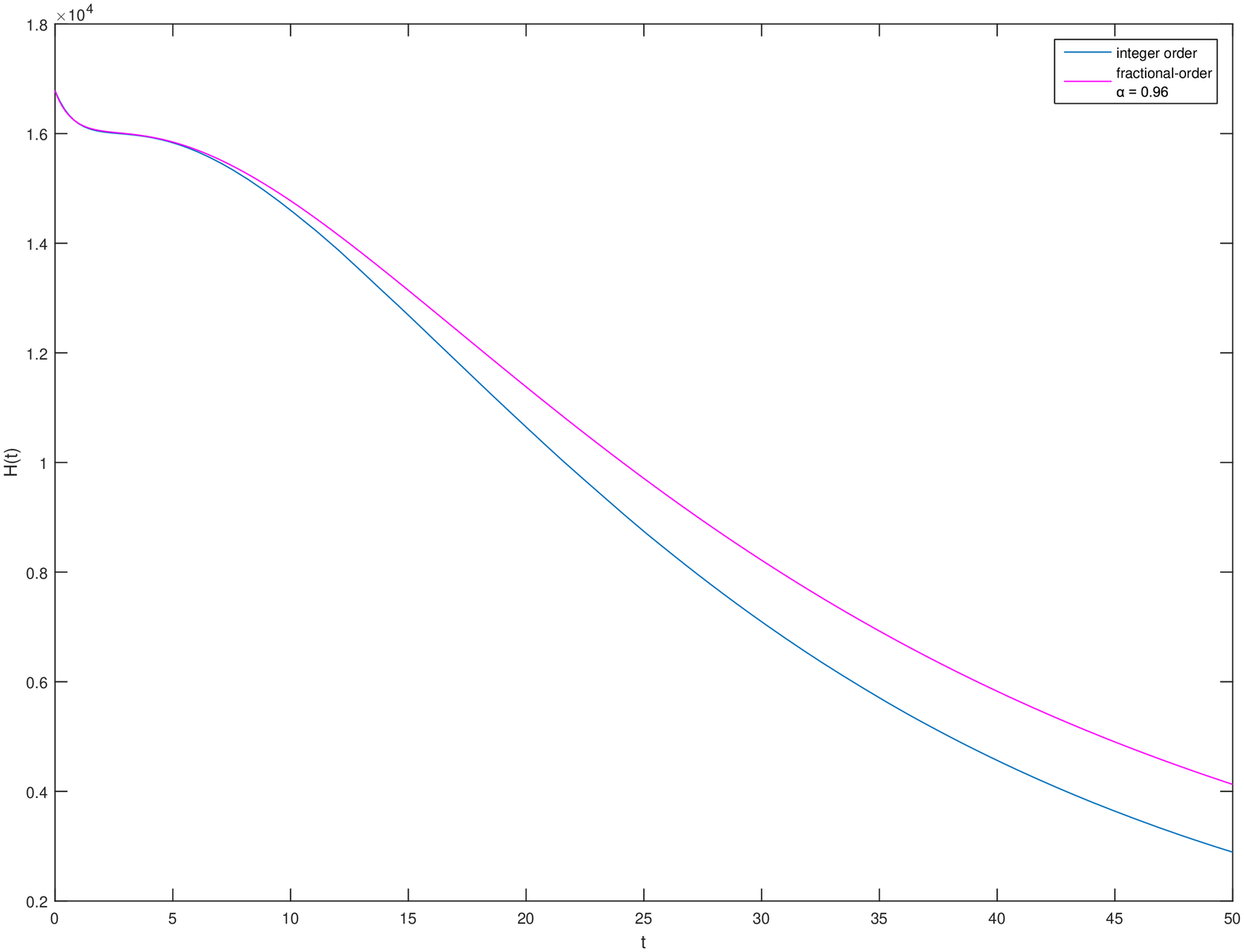}
		\caption{Graphical description of hospitalized cases versus time.}
		\label{f111}
	\end{subfigure}
	\quad 
	\begin{subfigure}[hb]{0.7\linewidth}
		\includegraphics[width=\linewidth]{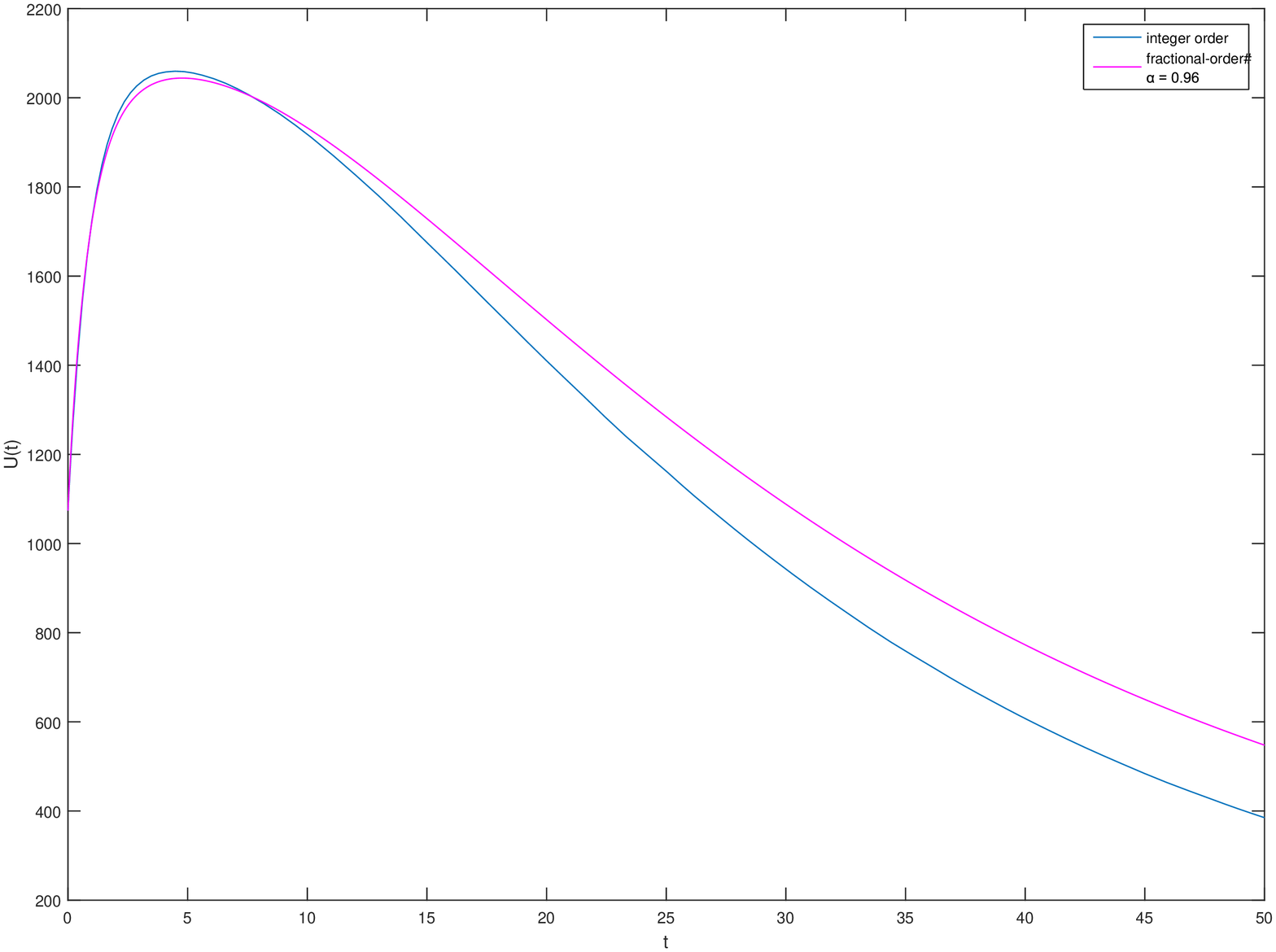}
		\caption{Graphical description of critical care cases versus time.}
		\label{f112}
	\end{subfigure}
	\caption{Simulation results of integer-order and fractional-order SIHUR model for France.}
	\label{global fig101}
\end{figure}


\begin{figure} [h]
	\centering
	\begin{subfigure}[ht]{0.7\linewidth}
		\includegraphics[width=\linewidth]{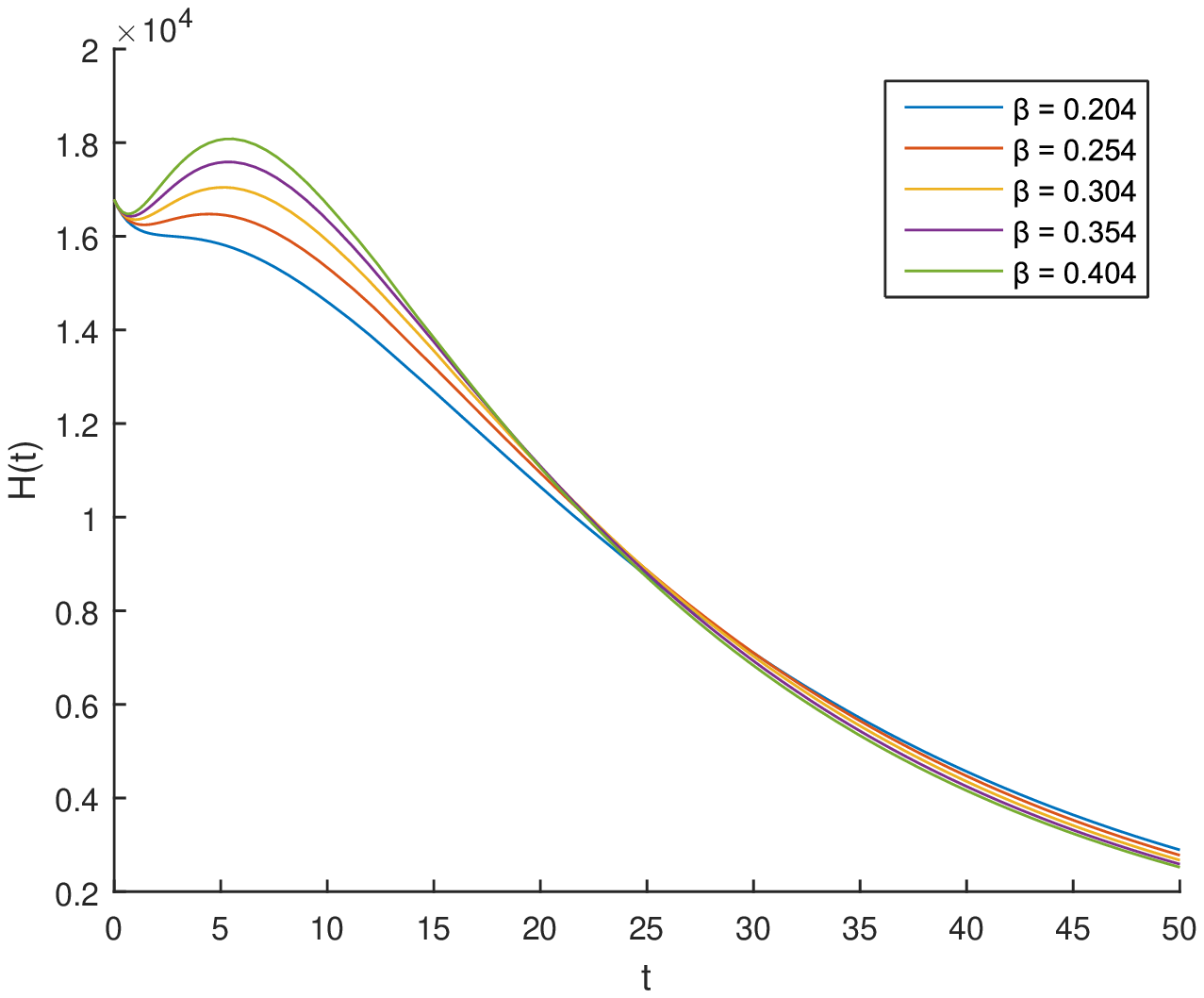}
		\caption{Graphical description of precautionary measures mitigation ($\beta$) on hospitalized cases of COVID-19.}
		\label{f116}
	\end{subfigure}

	\begin{subfigure}[hb]{0.7\linewidth}
		\includegraphics[width=\linewidth]{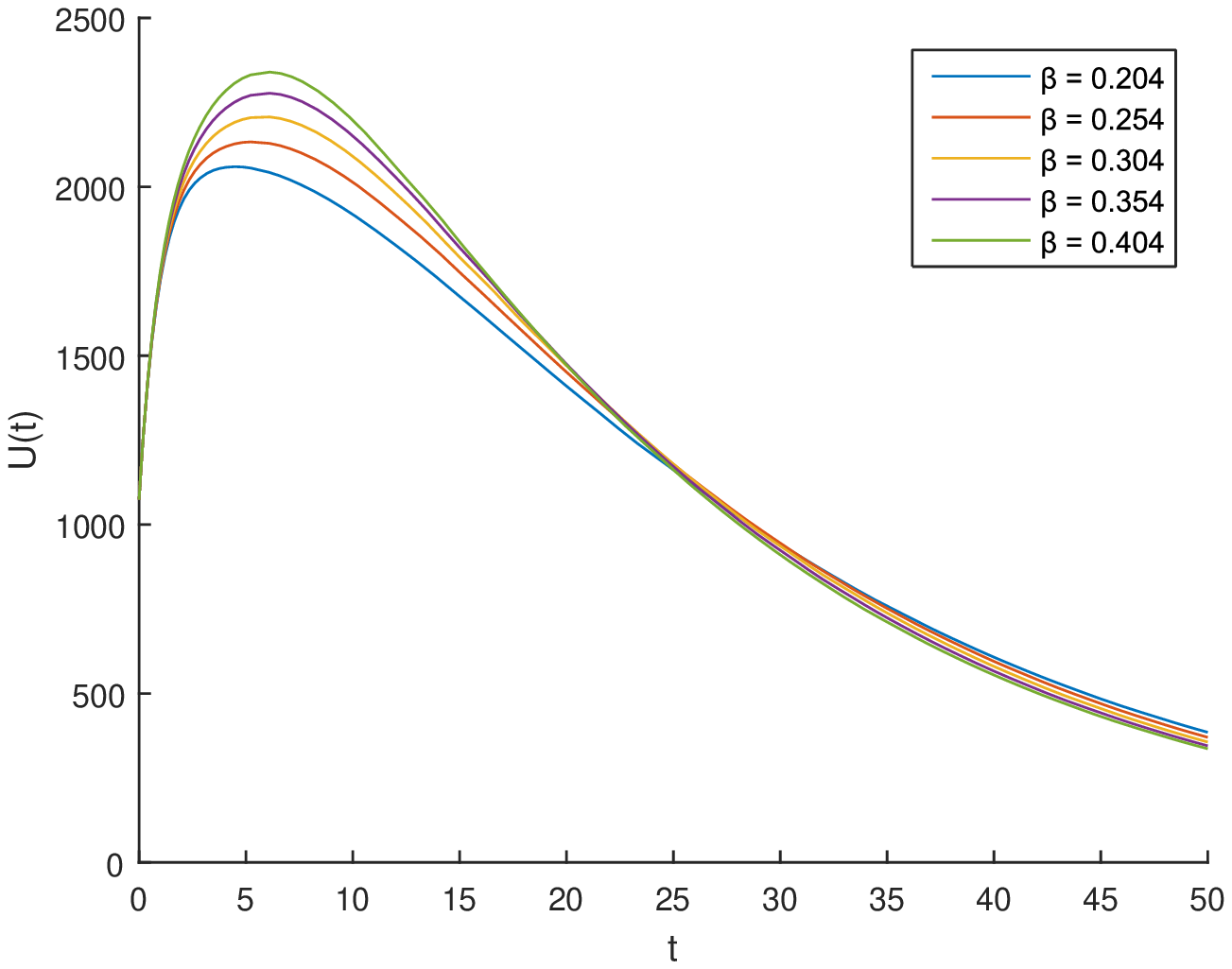}
		\caption{Graphical description of precautionary measures mitigation ($\beta$) on critical care cases of COVID-19.}
		\label{f117}
	\end{subfigure}
	\caption{Impact of precautionary measures mitigation ($\beta$) on hospitalized and ICUs' cases of COVID-19 for France by integer order SIHUR model.}
	\label{global fig103}
\end{figure}


\begin{figure} [h]
	\centering
	\begin{subfigure}[ht]{0.7\linewidth}
		\includegraphics[width=\linewidth]{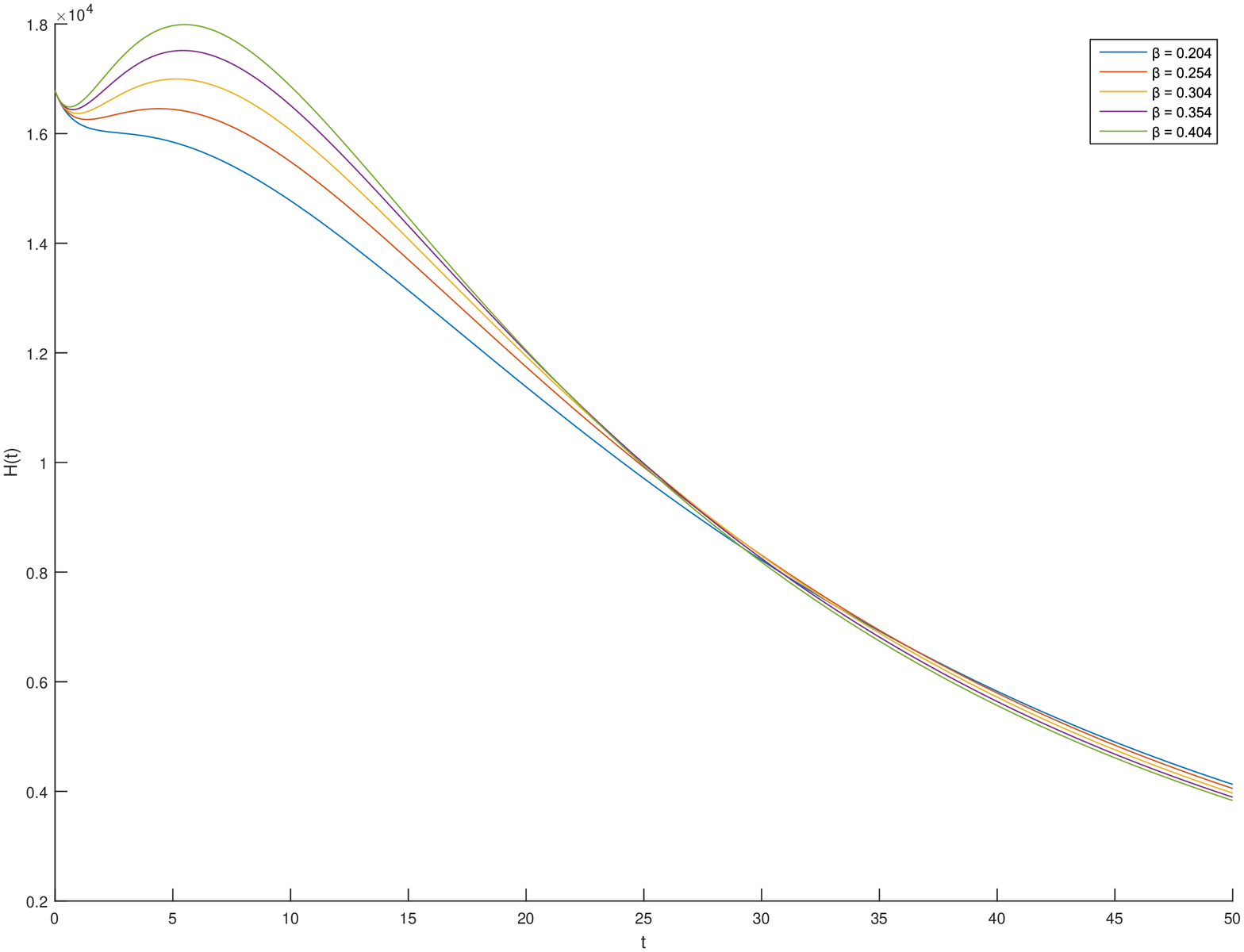}
		\caption{Graphical description of precautionary measures mitigation ($\beta$) on hospitalized cases of COVID-19.}
		\label{f118}
	\end{subfigure}
	
	\begin{subfigure}[hb]{0.7\linewidth}
		\includegraphics[width=\linewidth]{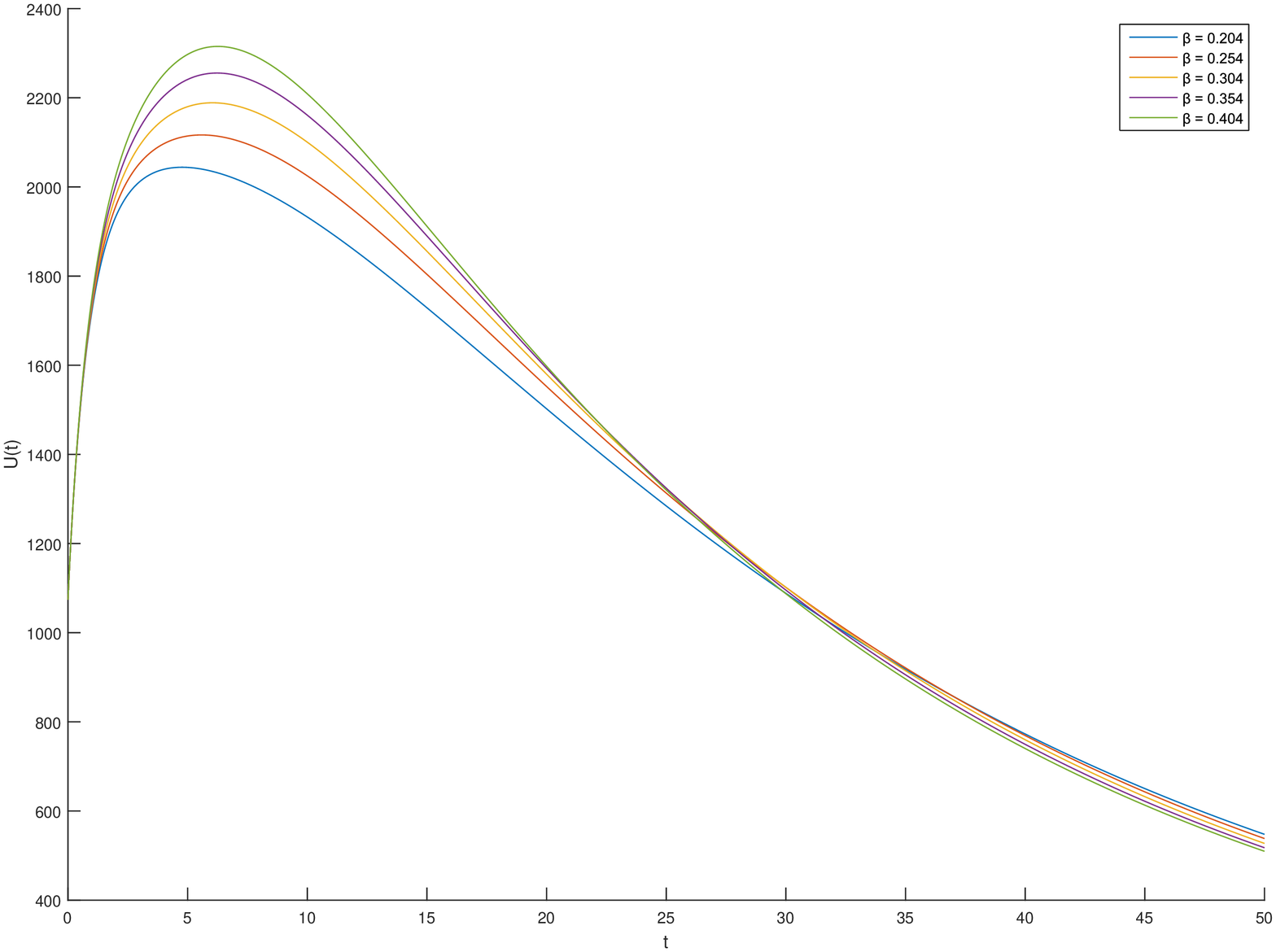}
		\caption{Graphical description of precautionary measures mitigation ($\beta$) on critical care cases of COVID-19.}
		\label{f119}
	\end{subfigure}
	\caption{Impact of precautionary measures mitigation ($\beta$) on hospitalized and ICUs' cases of COVID-19 for France by fractional-order SIHUR model.}
	\label{global fig104}
\end{figure}

\section{Conclusion}\label{113}
In this work, we investigate the problem of COVID-19 breaking out in the entire French hospitals and ICUs. The integer order and fractional-order models are proposed to predict the COVID-19 spread in France from $24^{th}$ May 2022 to $13^{th}$ July 2022. In this paper, with the SIHUR model, we also study the impact of the precautionary measures. We find that the increase in the total number of hospitalized and critical care cases is in the first days of our study. And this is followed by a significant decrease in the cumulative number of COVID-19 cases. Clearly, we show that implementing the precautionary measures has an apparent minor effect on the total number of patients in hospitals and ICUs.


\bibliographystyle{model1-num-names}
\bibliography{ref}
\end{document}